\documentclass[11pt,a4paper]{article}
\usepackage[margin=2.5cm]{geometry}

\usepackage[english]{babel}
\usepackage[utf8]{inputenc}

\usepackage[T1]{fontenc}
\usepackage{amsmath,amsthm,amssymb}

\usepackage[libertine]{newtxmath}
\usepackage{libertine}
\usepackage[scaled=0.75]{beramono}

\usepackage{graphicx}
\usepackage[colorinlistoftodos]{todonotes}
\usepackage{hyperref}
\usepackage{amsfonts}
\usepackage{mathrsfs}
\usepackage{mathtools}
\usepackage{verbatim}
\usepackage{footnote}
\usepackage{algorithm}
\usepackage[noend]{algpseudocode}
\usepackage{lineno}
\usepackage{caption}
\usepackage{tikzsymbols}
\usepackage[capitalize, nameinlink]{cleveref}
\hypersetup{colorlinks={true},linkcolor={blue},citecolor=cyan}
\usepackage{tikz}
\usetikzlibrary{arrows.meta}
\usepackage{boxedminipage2e}

\usepackage{eqparbox,array}

\usepackage{todonotes}

\newtheorem{theorem}{Theorem}
\newtheorem{lemma}[theorem]{Lemma}
\newtheorem{corollary}[theorem]{Corollary}

\newtheorem{claim}[theorem]{Claim}

\newtheorem*{remark}{Remark}

\renewcommand{\ge}{\geqslant}
\renewcommand{\le}{\leqslant}

\newcommand{\eps}{\varepsilon}

\DeclareMathOperator{\poly}{poly}
\DeclareMathOperator{\polylog}{polylog}
\DeclareMathOperator{\alg}{ALG}

\DeclareMathOperator{\RR}{\mathbb{R}}

\newcommand{\cW}{\mathcal{W}}

\newcommand{\errpr}{\lambda}
\newcommand{\clog}[1]{\lceil \log #1 \rceil}
\newcommand{\ceile}[1]{\lceil #1 \rceil}
\newcommand{\floore}[1]{\lfloor #1 \rfloor}
\newcommand{\ceilelr}[1]{\left\lceil #1 \right\rceil}
\newcommand{\floorelr}[1]{\left\lfloor #1 \right\rfloor}

\DeclareMathOperator{\dimn}{dim}

\newcommand{\indx}{\textsc{index}\xspace}

\begin{document}
\title{Fully-Dynamic Coresets}
\author{
  Monika Henzinger\footnote{The research leading to these results has received
    funding from the European Research Council under the European Union's Seventh
    Framework Programme (FP/2007-2013) / ERC Grant Agreement no. 340506.} \\
  University of Vienna,\\
  Faculty of Computer Science\\
  \texttt{monika.henzinger@univie.ac.at}
  \and
  Sagar Kale\footnote{Fully supported by Vienna Science and Technology Fund
    (WWTF) through project ICT15-003.} \\
  University of Vienna,\\
  Faculty of Computer Science\\
  \texttt{sagar.kale@univie.ac.at}
}

\date{}

\maketitle

\begin{abstract}
  With input sizes becoming massive, coresets---small yet representative summary
  of the input---are relevant more than ever.  A weighted set $C_w$ that is a
  subset of the input is an $\eps$-coreset if the cost of any feasible solution
  $S$ with respect to $C_w$ is within $[1 {\pm} \eps]$ of the cost of $S$ with
  respect to the original input.  We give a very general technique to compute
  coresets in the fully-dynamic setting where input points can be added or
  deleted.  Given a static (i.e., not dynamic) $\eps$-coreset-construction
  algorithm that runs in time $t(n, \eps, \errpr)$ and computes a coreset of
  size $s(n, \eps, \errpr)$, where $n$ is the number of input points and
  $1 {-}\errpr$ is the success probability, we give a fully-dynamic algorithm
  that computes an $\eps$-coreset with worst-case update time
  $O((\log n) \cdot t(s(n, \eps/\log n, \errpr/n), \eps/\log n, \errpr/n) )$
  (this bound is stated informally), where the success probability is
  $1{-}\errpr$.  Our technique is a fully-dynamic analog of the merge-and-reduce
  technique, which is due to Har-Peled and Mazumdar~\cite{har_peled_mazumdar04}
  and is based on a technique of Bentley and Saxe~\cite{bentley_saxe80}, that
  applies to the insertion-only setting where points can only be added.
  Although, our space usage is $O(n)$, our technique works in the presence of an
  adaptive adversary, and we show that $\Omega(n)$ space is required when
  adversary is adaptive.

  As a concrete implication of our technique, using the result of Braverman et
  al.~\cite{braverman_etal_16}, we get fully-dynamic $\eps$-coreset-construction
  algorithms for $k$-median and $k$-means with worst-case update time
  $O(\eps^{-2}k^2\log^5 n \log^3 k)$ and coreset size
  $O(\eps^{-2}k\log n \log^2 k)$ ignoring $\log \log n$ and $\log(1/\eps)$
  factors and assuming that $\eps = \Omega(1/\poly(n))$ and
  $\errpr = \Omega(1/\poly(n))$ (which are very weak assumptions made only to
  make these bounds easy to parse).  This results in the first fully-dynamic
  constant-approximation algorithms for $k$-median and $k$-means with update
  times $O(\poly(k, \log n, \eps^{-1}))$.  The dependence on $k$, specifically,
  is only quadratic, and the bounds are worst-case.  The best previous bound for
  both problems was \emph{amortized} $O(n\log n)$ by Cohen-Addad et
  al.~\cite{cohen_addad_etal19} via randomized $O(1)$-coresets in $O(n)$ space.

  We also show that under the OMv conjecture~\cite{HenzingerKNS15}, a
  fully-dynamic $(4 - \delta)$-approximation algorithm for $k$-means must either
  have an amortized update time of $\Omega(k^{1-\gamma})$ or amortized query
  time of $\Omega(k^{2 - \gamma})$, where $\gamma > 0$ is a constant.
\end{abstract}


\thispagestyle{empty}
\newpage
\clearpage
\setcounter{page}{1}

\section{Introduction}
\label{sec:intro}
Clustering is an ubiquitous notion that one encounters in computer-science areas
such as data mining, machine learning, image analysis, bioinformatics, data
compression, and computer graphics, and also in the fields of medicine, social
science, marketing, etc.  Today, when the input data has become massive, one
would rather run an algorithm on a small but representative summary of the
input, and for clustering problems, a \emph{coreset} serves that function
perfectly.  The concept of a coreset was defined first in computational geometry
as a small subset of a point set that approximates the shape of the point set.
The word coreset has now evolved to mean an appropriately weighted subset of the
input that approximates the original input with respect to solving a
computational problem.

Let $P$ be a problem for which the input is a weighted subset\footnote{When
  using a set operation such as union or notation such as $\subseteq$ with one
  or more weighted sets, we mean it for the underlying unweighted sets.  Also,
  all weights are nonnegative.}  $X_w \subseteq U$; think of $U$ as in a metric
space $(U, d)$, so $U$ is unweighted and $d$ is the distance function.  Let
$n := |X_w|$ and $W := \sum_{x \in X_w}w(x)$.  We also refer to elements of $U$
as points.  The goal in the problem $P$ then is to output $S^*$ that belongs to
the feasible-solution space (or \emph{query} space) $Q$ such that the cost
$c(S^*, X_w)$ is minimized.  For example, in the \emph{$k$-median}
(respectively, \emph{$k$-means}) problem, $Q$ is the set of all (unweighted)
subsets of $X_w$ of cardinality at most $k$ and
$c(S, X_w) := \sum_{x \in X_w} w(x)\min_{s \in S}d(x, s)$ (respectively,
$\sum_{x \in X_w} w(x)\min_{s \in S}(d(x, s))^2$).  Then, for the problem $P$, a
weighted set $C_w$ such that $C_w \subseteq X_w$ is an \emph{$\eps$-coreset} if,
for any feasible solution $S \in Q$, we have that
$c(S, X_w) \in [1 {\pm} \eps] c(S, C_w)$; we sometimes say that the
\emph{quality} of coreset $C_w$ is $\eps$.  For many problems, fast
coreset-construction algorithms exist; e.g., for $k$-median and $k$-means,
$\tilde{O}(nk)$-time\footnote{Logarithmic factors are hidden in the $\tilde{O}$
  notation.} algorithms for computing $\eps$-coresets of size
$O(\eps^{-2}k\polylog(n))$ exist.

Throughout the paper, we assume that the cost function $c$ for the problem $P$
is \emph{linear}: for any weighted subsets $Y^1_w, Y^2_w \subseteq U$ with
disjoint supports and any $S \in Q$, we have that
$c(S, Y^1_w \cup Y^2_w) = c(S, Y^1_w) + c(S, Y^2_w)$, where the union
$Y^1_w \cup Y^2_w$ is the weighted union.  It is easy to see that $k$-median and
$k$-means cost functions are linear.

Our goal in this paper is to give \emph{dynamic algorithms} for computing a
coreset.  In the dynamic setting, the input changes over time.  A dynamic
algorithm is a data structure that supports three types of operations:
\emph{Insert$(p, w)$}, which inserts a point $p$ with weight $w$ into $X_w$;
\emph{Delete$(p)$}, which removes point $p$ from $X_w$; and \emph{Query$()$},
which outputs a coreset of $X_w$.  Weight updates can be simulated by deleting
and re-inserting a given point, or the data structure may support a separate
weight-changing operation.  This is known as the \emph{fully} dynamic model as
opposed to the \emph{insertion-only} setting where a point can only be inserted.
At any time instant, a coreset is maintained by the algorithm, and the
complexity measure of interest is the \emph{update} time, i.e., how fast the
solution can be updated after receiving a point update, and also the \emph{size}
of the coreset, which determines the \emph{query} time.  Suppose there is a
dynamic coreset-construction algorithm, say $\alg_D$, for a problem $P$. Then a
solution for the problem $P$ can be maintained dynamically by running $\alg_D$,
and on query, a solution is computed by querying $\alg_D$ and running a static
(i.e., not dynamic) algorithm for $P$ on the returned coreset.  In this paper,
we give a very general technique on how to maintain a coreset in the
fully-dynamic setting: given a \emph{static} coreset-construction algorithm for
any problem $P$, we show how to turn it into a dynamic coreset-construction
algorithm for $P$.

Intuitively, our technique is to the fully-dynamic setting as the
merge-and-reduce technique is to the insertion-only setting.  The
\emph{merge-and-reduce} technique, which is based on a technique of Bentley and
Saxe~\cite{bentley_saxe80}, is due to Har-Peled and
Mazumdar~\cite{har_peled_mazumdar04} and is a fundamental technique to obtain an
insertion-only coreset-construction algorithm using a static
coreset-construction algorithm, say $\alg_S$, as a black box.  Loosely speaking,
it is as follows.  At any time instant, the algorithm maintains up to $\clog{n}$
buckets.  For $i \in \{1, 2, \ldots, \clog{n}\}$, the bucket $B_i$ has capacity
$2^{i-1}$, each bucket can be either \emph{full}, (i.e. at capacity $2^{i-1}$)
or \emph{empty}, and each point goes in exactly one bucket.  Then at any
time-instant, the current number of points uniquely determines the states of the
buckets.  Whenever a point is inserted, the states of the buckets change like a
binary counter.  That is, the new point goes into $B_i$, where $B_i$ is the
smallest-index empty bucket, and all the points in $\cup_{j=1}^{i-1} B_j$ are
moved to $B_i$ (\emph{merge}). Note that this creates a full bucket $B_i$.  Then
a coreset is computed on $B_i$ by running $\alg_S$ on it (\emph{reduce)}.  The
overall coreset is then just union of all non-empty buckets.

We show that a similar result can be achieved in the fully-dynamic setting.  Our
main result is the following theorem (stated slightly informally).
\begin{theorem}
  Assume that there is a static coreset construction algorithm for a problem $P$
  with linear cost function that \textbf{a)}~runs in time
  $t_P(n_s, \eps_s, \errpr_s, W_s)$, \textbf{b)}~always outputs a set of
  cardinality at most $s_P(\eps_s, \errpr_s, W_s)$ and total weight at most
  $(1{+} \delta)W_s$, and \textbf{c)}~has the guarantee that the output is an
  $\eps_s$-coreset with probability at least $1{-} \errpr_s$, where $n_s$ is the
  number of \emph{integer}-weighted input points and $W_s$ is the total weight
  of points.

  Then there is a fully-dynamic coreset-construction algorithm for $P$ that,
  with \emph{rational}-weighted input points, \textbf{a)}~always maintains an
  output set of cardinality at most $s_P(\eps, \errpr, W)$, \textbf{b)}~has the
  guarantee that the output is an $\eps$-coreset with probability at least
  $1{-} \errpr$, and \textbf{c)}~has worst-case update time
  \[
    O\left((\log n) \cdot t_P\left(s_P^*, \frac{\eps}{\log n}, \frac{\errpr}{n},
        W\right)\right)\,,
\]
where $n$ is the current number of points,
$W = O((1{+} \delta)^{\clog{n}}\poly(n))$, and
$s_P^* = s_P\left(\frac{\eps}{\log n}, \frac{\errpr}{n}, W\right)$.
\end{theorem}

We mention below a concrete implication of the above theorem for $k$-median and
$k$-means using the result of Braverman et al.~\cite{braverman_etal_16}.
\begin{theorem}
  \label{thm:kmedkmeans}
  For the $k$-median and $k$-means problems, there is a fully-dynamic algorithm
  that maintains a set of cardinality
  $O(\eps^{-2}k(\log n \log k \log(k\eps^{-1}\log n) + \log (1/\errpr)))$, that
  is an $\eps$-coreset with probability at least $1{-} \errpr$, and has
  worst-case update time
  $O\left(\eps^{-2} k^2\log^5n\log^3 k \log^2(1/\eps) (\log\log n)^3\right)$,
  assuming that $\eps = \Omega(1/\poly(n))$ and $\errpr = \Omega(1/\poly(n))$.
  \footnote{We make these very weak assumptions to simplify some extremely
    unhandy factors involving $\eps$ and $\errpr$ in the expression for the
    update time.}
\end{theorem}
Ignoring $\log \log n$ and $\log(1/\eps)$ above, the coreset cardinality is
$O(\eps^{-2}k\log n \log^2 k)$ and worst-case update time is
$O(\eps^{-2}k^2\log^5 n \log^3 k)$.  It can be easily proved that running an
$\alpha$-approximation algorithm for $k$-median on an $\eps$-coreset gives a
$2\alpha(1{+}\eps)$-approximation whereas that for $k$-means gives a
$4\alpha(1{+}\eps)$-approximation.  Any such polynomial-time static
algorithm---say, e.g., $(5+\eps')$-approximation algorithm for $k$-median by
Arya et al.~\cite{aryagkmmp_04}
and $16$-approximation algorithm for $k$-means by Gupta and
Tangwongsan~\cite{guptat_08}---can be run on our output coreset in
$O(\poly(k, \log n, \eps^{-1}))$ time to obtain a constant approximation.  This
is the first fully-dynamic constant-approximation algorithm for $k$-median and
$k$-means whose \emph{worst-case} time per operation is polynomial in $k$,
$\log n$, and $\eps^{-1}$.  The best previous result was a randomized algorithm
with \emph{amortized} $O(n\log n)$ update time and $O(n)$ space by Cohen-Addad
et al.~\cite{cohen_addad_etal19}.

With a simple reduction, we also show a conditional lower bound on the time per
operation for $k$-means.  The following theorem is proved as
Theorem~\ref{thm:lbomv} in Section~\ref{sec:lower}.
\begin{theorem}
  Let $\gamma > 0$ be a constant.  Under the OMv
  conjecture~\cite{HenzingerKNS15}, for any $\delta > 0$, there does not exist a
  fully-dynamic algorithm that maintains a $(4 - \delta)$-approximation for
  $k$-means with amortized update time
  $O(k^{1-\gamma})$ and query time $O(k^{2-\gamma})$ such that over a polynomial
  number of updates, the error probability is at most $1/3$.
\end{theorem}

\paragraph{Our technique}
At the core, our technique is simple.  We always maintain a balanced binary tree
of depth $\clog{n}$ containing exactly $n$ leaf nodes (recall that $n$ is the
current number of points).  Each node corresponds to a subset of $X_w$, the
current input: each leaf node corresponds to a singleton (hence $n$ leaf nodes),
and an internal node corresponds to the weighted union of the sets represented
by its children.  If the cardinality of the union exceeds a certain threshold,
then we use the static coreset-construction algorithm to compute its coreset.
The root gives a coreset of the whole input.

We next explain how we handle updates in this data structure.  Point
insertions are straightforward: create a new leaf node and run all the
static-algorithm instances at the nodes on the leaf-to-root path. The way we
handle point deletions is similar in spirit to the way delete-min works in a
min-heap data structure: whenever a point at leaf-node $\ell_d$ is deleted, we
swap contents of $\ell_d$ with those of the rightmost leaf-node, say $\ell_r$,
and delete $\ell_r$, thus maintaining the balance of the tree.  Then we run all
the static-algorithm instances at the nodes on the two affected leaf-to-root
paths.

However, there are some complications that require new techniques to make it
work in \emph{worst-case} time.  To maintain guarantees for the output coreset
quality and overall success probability, we need to adapt the parameters
$\eps_s$ and $\errpr_s$ used for the static algorithm at the internal nodes.
The problem is that both depend on $n$, which changes over time and thus might
become outdated.  To show an \emph{amortized} update-time bound, we can simply
rerun the static algorithms at all internal nodes whenever $n$ has changed by a
constant factor.  To achieve our worst-case bound, we use two \emph{refresh
  pointers} that point at leaf nodes, and after each update operation, we rerun
using the new values of $\eps_s$ and $\errpr_s$ all the static-algorithm
instances at the nodes on the leaf-to-root path from the leaf nodes pointed to
by the refresh pointers.  This keeps the outputs of the static-algorithm
instances at the internal nodes always fresh.  After every update, we move these
pointers to the right so that they point to the next leaf nodes.

Further complications are caused by fractional weights at the leaf nodes and
fractional intermediate-output weights.  A problem arises when the weights in
$X_w$ are fractional, and the static algorithm expects integer-weighted
input~\cite{chen_09}.  Even if the static algorithm can handle fractional
weights~\cite{feldmanl_11,braverman_etal_16}, there can be a problem.  The
output of the static algorithm at an internal node is the input for the static
algorithm at its parent.  Na{\"i}vely feeding these output fractional weights
directly to the static algorithm at the parent may result in numbers exponential
in $n$ near the root, thus prohibitively increasing the update time.  To deal
with these problems, rounding is needed for the input, i.e., at the leaf nodes,
as well as for each intermediate-output at an internal node.  Thus, we propose a
more sophisticated rounding scheme and show that the rounding errors accumulated
by our rounding are not too high.

We note that our balanced binary-tree data structure may be used to get dynamic
algorithms in the following situations.  Let
$f:\RR^{\dimn} \rightarrow \RR^{\dimn}$ be a \emph{multi-valued} function.
Suppose for any $u$ and $v$ with disjoint supports and for any $f_u \in f(u)$
and $f_v \in f(v)$, we have $f_u + f_v \in f(u + v)$.  Also suppose that
$f(f(v)) \subseteq f(v)$ for any $v$.  Now, given input $v$, we want to compute
some vector in $f(v)$.  If there is a static algorithm for this, then using our
technique, we can maintain some vector in $f(v)$ for a dynamically changing
vector $v$.  The allowed dynamic operation on $v$ is ``add $a$ to the $i$th
coordinate of $v$,'' where $a \in \RR$.  The resulting dynamic algorithm is fast
if the static algorithm always outputs a ``small'' vector; this is true for
coresets because coresets are small by nature.  Thinking about coresets in the
above language, each point is an identity vector in $\RR_+^{|U|}$, and then each
weighted set of points naturally identifies with a vector.  An $\eps$-coreset
reduces the number of points drastically.  Union of coresets of two disjoint
sets is a coreset of the union of those two sets (see
Lemma~\ref{lem:coresetunion}).  Although an $\eps$-coreset of an $\eps$-coreset
is not an $\eps$-coreset, it is a $(2\eps + \eps^2)$-coreset (see
Lemma~\ref{lem:coresetcomposition}).

\paragraph{Space}
In the merge-and-reduce technique, a bucket $B_i$ will not actually contain
$2^{i-1}$ points but just a coreset of $2^{i-1}$ points that would have been
there otherwise at any time instant.  Thus, using space just $\clog{n}$ times
the coreset size for a bucket, one can get a coreset of the whole
input~\cite{har_peled_mazumdar04}.  This makes it also applicable in the more
restricted streaming model, where the input points arrive in a sequence and the
goal is to compute a coreset using sublinear space.  In the fully-dynamic
setting, deletions also need to be handled, and hence no deterministic or
randomized algorithm against an adaptive adversary that stores only a coreset is
possible: the adversary generating the input could simply ask a query and then
delete all points in the returned coreset.  Hence, an algorithm that does not
store any information about the non-coreset points would not be able to maintain
a valid coreset.  Even though we store all the points in our fully-dynamic
technique, i.e., its space usage is $O(n)$, it works against an adaptive
adversary because we never make any assumption about the next update and perform
each update independently of all previous updates.  By a straightforward
reduction from the communication problem of $\indx$, we show that $\Omega(n)$
space is required in the presence of an adaptive adversary.  The proof of the
following theorem appears in Section~\ref{sec:lower}.
\begin{theorem}
  A fully-dynamic algorithm that obtains any bounded approximation for
  $1$-median or $1$-means that works in the presence of an adaptive adversary
  and has success probability $1 - 1/(8n^2)$ must use $\Omega(n)$ space, where
  $n$ is the current number of points.
\end{theorem}

\paragraph{Comparison to prior work}
Our technique is close to the sparsification technique of Eppstein et
al.~\cite{eppstein_etal97} that is used to speed up dynamic graph algorithms.
There, one has to assume that the number of vertices in the input graph, say
$n_v$, does not change, but the edge set changes dynamically, and the bounds are
obtained in terms of $n_v$ and $m$, the current number of edges.  Their dynamic
edge-tree structure is based on a fixed vertex-partition tree.  In the
vertex-partition tree, a node at level~$i$ corresponds to a vertex-set of
cardinality $n_v/2^i$, and a vertex-set at a node is a union of its children's
vertex sets (cf. our technique).  To start using the edge tree, the
vertex-partition tree has to be built first and hence the knowledge of $n_v$ is
necessary.  Neither do we need such a fixed structure nor any information about
the number of points.  Also, in the sparsification technique, there is no analog
of weight handling/rounding.  Another crucial difference is that they do not use
a routine analogous to our refresh-pointers routine because their internal-node
guarantees are always fresh.  As we discussed before, these refresh pointers are
critical for us also in making sure that the error introduced by the unavoidable
rounding of output weights of the static-algorithm instances is kept in check.

Ali Mehrabi recently pointed out to us that the paper by  Har-Peled and Mazumdar~\cite{har_peled_mazumdar04} gives a fully-dynamic algorithm that maintains coresets.  Specifically, they state
that for $k$-median and $k$-means in Euclidean metric in $\RR^{\dimn}$, there
exists an algorithm to compute an $\eps$-coreset with amortized update time
$O(k \eps^{-\dimn} \log^{\dimn + 2} n \log (k \log n/\eps) + k^5 \log^{10}n)$.
This result is based on the dynamization technique of Agarwal et
al.~\cite{agarwal_etal_04}; our technique can be seen as an extension of this
dynamization technique.  They also maintain a balanced binary tree, but this
tree is rebuilt after a constant fraction of updates. Thus, their running time guarantees
are only amortized, whereas we crucially use refresh pointers to get the
worst-case bounds.  Furthermore, as described, Agarwal et al.'s technique works only when the
input is integer weighted and the static algorithm outputs an integer weighted
coreset.

\subsection{Further related Work}
\label{sec:related-work}

For $k$-median and $k$-means, the first coreset-construction algorithms were by
Har-Peled and Mazu\-mdar~\cite{har_peled_mazumdar04} for Euclidean metrics and
by Chen~\cite{chen_09} for general metrics.  Improved algorithms computing
smaller coresets were subsequently obtained by Har-Peled and
Kushal~\cite{har_peledk_07} and by Feldman and Langberg~\cite{feldmanl_11}.  The
current known best is by Braverman et al.~\cite{braverman_etal_16}:
$O(\eps^{-2}k\log k \log n)$-size coresets in $\tilde{O}(nk)$ time, who also
give an excellent summary of the literature on coresets that we highly
recommend.  Note that by merge-and-reduce technique, each improvement also gave
rise to better (insertion-only) streaming coreset constructions.  For $k$-median
and $k$-means, Frahling and Sohler~\cite{frahlings_05} gave the first
coreset-construction algorithm in the dynamic-streaming setting where points can
be added or removed.  It uses space and update time of
$O(\poly(\eps^{-1}, \log m, \log \Delta))$ for constant $k$ and $\dimn$ when the
points lie in the discrete Euclidean metric space $\{1,\ldots,\Delta\}^{\dimn}$;
for $k$-median, this was recently improved to
$O(\eps^{-2}k\poly(\dimn, \log \Delta))$ space and update time of
$O(\poly(\eps^{-1}, k, \dimn, \log \Delta))$ by Braverman et
al.~\cite{bravermanflsy_17}.  Coreset constructions with improvements in certain
parameters in the Euclidean settings have been
obtained~\cite{feldmanss_13,sohlerw_18}.

The $k$-median and $k$-means problems have received significant attention in the
algorithms community~\cite{charikargts_02, jainv_01, jainms_02, charikarg_05,
  aryagkmmp_04, mettup_04, kanungomnpsw_04, guptat_08, lis_16, ahmadiannsw_17,
  byrkaprst_17}.  The best approximation ratio for $k$-median is $2.675+\eps$ by
Byrka et al.~\cite{byrkaprst_17} and that for $k$-means is $9+\eps$ by Ahmadian
et al.~\cite{ahmadiannsw_17}.


\section{Preliminaries}
Let us fix a problem $P$ with the input $X_w$, the set of feasible solutions
$Q$, and the linear cost function $c: Q \times \cW \rightarrow \RR_+$, where
$\cW$ is the set of all weighted subsets\footnote{To be precise: denote
  unweighted version of $X_w$ by $X'$, then $\cW$ is essentially $\RR^{X'}_+$.}
of $X_w$.  All the numbers encountered are nonnegative.

\paragraph{The computational model}
The input set $X_w$ is a weighted set of $n$ points having rational weights
whose numerators and denominators are bounded by $O(\poly(n))$.  The algorithm
works in the random access machine model with word size $O(\log n)$.  Each
memory word can be accessed in constant time.  With each update, a new point is
inserted, an existing point is deleted, or the weight of an existing point is
modified by adding or subtracting a nonnegative number.  The net weight of each
point always stays nonnegative with its numerator and denominator always bounded
by $O(\poly(n))$.

We will prove some basic lemmas about coresets.  Using these, we can take
weighted union of two coresets without any loss (Lemma~\ref{lem:coresetunion})
and take a coreset of a coreset without much loss
(Lemma~\ref{lem:coresetcomposition}).
\begin{lemma}
  \label{lem:coresetunion}
  If $C^1_w$ and $C^2_w$ are $\eps$-coresets of $X_w^1$ and $X_w^2$,
  respectively, with respect to a linear cost function $c$ such that
  $X_w^1 \cap X_w^2 = \emptyset$, then $C^1_w \cup C^2_w$ is an $\eps$-coreset
  of $X_w^1 \cup X_w^2$.
\end{lemma}
\begin{proof}
  By linearity of $c$: for any $S \in Q$, 
  \[
    c(S, X_w^1 \cup X_w^2) = c(S, X_w^1) + c(S, X_w^2) \in [1 {\pm} \eps]
    \left(c(S, C^1_w) + c(S, C^2_w)\right) = [1 {\pm} \eps] c(S, C^1_w\cup
    C^2_w)\,,
  \]
  where, recall that, $C^1_w\cup C^2_w$ is a weighted union.
\end{proof}

\begin{lemma}
  \label{lem:coresetcomposition}
  If $C'_w$ is an $\eps$-coreset of $C_w$, and $C''_w$ is a
  $\delta$-coreset of $C'_w$, both with respect to $c$, then $C''_w$ is an
  $(\eps + \delta + \eps\delta)$-coreset of $C_w$ with respect to $c$.
\end{lemma}
\begin{proof}
  For any $S \in Q$, we have
  $c(S, C_w) \in [1 {\pm} \eps] c(S, C'_w)$ and
  $c(S, C'_w) \in [1 {\pm} \delta] c(S, C''_w)$.  So,
  \[
    c(S, C_w) \ge (1{-} \eps) c(S, C'_w) \ge (1{-} \eps) (1{-} \delta) c(S, C''_w)
    = (1 - \eps - \delta + \eps\delta) c(S, C''_w)
    \ge (1 - \eps - \delta - \eps\delta) c(S, C''_w)\,,
  \] and
  $
    c(S, C_w) \le (1{+} \eps) c(S, C'_w) \le (1{+} \eps) (1{+} \delta) c(S, C''_w)
    = (1 + \eps + \delta + \eps\delta) c(S, C''_w)
  $.
\end{proof}

Let $C^1_w$ be an $\eps$-coreset of $C_w$ and $C^2_w$ be an
$\eps$-coreset of $C^1_w$.  Then we say that $C^1_w$ and $C^2_w$ are,
respectively, $1$-level and $2$-level $\eps$-coresets of $C_w$.  Extending
this notion, we define an $i$-level $\eps$-coreset to 
be an $\eps$-coreset of an $(i-1)$-level $\eps$-coreset.

\begin{lemma}
  \label{lem:ilevelcoreset}
  If $C^\ell_w$ is an $\ell$-level $\eps$-coreset of $C_w$, then $C^\ell_w$ is a
  $\left(\sum_{i = 1}^\ell\binom{\ell}{i} \eps^i \right)$-coreset of $C_w$.
\end{lemma}
\begin{proof}
  The proof is by induction on $\ell$.  Base case is when $\ell = 1$, and by
  definition, a $1$-level coreset is an $\eps$-coreset.  By induction
  hypothesis, we have that $C^{\ell - 1}_w$ is a
  $\left(\sum_{i = 1}^{\ell - 1}\binom{\ell - 1}{i} \eps^i \right)$-coreset
  of $C_w$.  Now, $C^\ell_w$ is an $\eps$-coreset of $C^{\ell - 1}_w$, hence
  by Lemma~\ref{lem:coresetcomposition}, $C^\ell_w$ is an
  $\left(\eps + (1{+}\eps)\sum_{i = 1}^{\ell - 1}\binom{\ell - 1}{i} \eps^i
  \right)$-coreset of $C_w$.  Now, use Lemma~\ref{lem:binom}, which appears
  below, with $\alpha = \eps$ to finish the proof.
\end{proof}

We prove two basic lemmas.

\begin{lemma}
  \label{lem:binom}
  For any positive integer $\ell$ and $\alpha \in \RR_+$, we have
  $\alpha + (1{+} \alpha)\sum_{i = 1}^{\ell - 1}\binom{\ell - 1}{i} \alpha^i =
  \sum_{i = 1}^{\ell}\binom{\ell}{i} \alpha^i$.
\end{lemma}
\begin{proof}[Proof idea]
  The proof is provided in Appendix~\ref{app:lembinom} and uses elementary
  identities involving binomial coefficients and algebraic manipulations.
\end{proof}

\begin{lemma}
  \label{lem:binomgeom}
  For any positive integer $\ell$ and $\alpha \in [0, 1]$, we have
  $\sum_{i = 1}^\ell\binom{\ell}{i} \left(\frac{\alpha}{2\ell}\right)^i \le
  \alpha$.
\end{lemma}
\begin{proof}
  $
    \sum_{i = 1}^\ell\binom{\ell}{i} \left(\frac{\alpha}{2\ell}\right)^i \le
    \sum_{i = 1}^\ell\ell^i \frac{\alpha^i}{2^i \ell^i} = \sum_{i = 1}^\ell
    \frac{\alpha^i}{2^i} \le \sum_{i = 1}^\ell \frac{\alpha}{2^i} \le \alpha
  $.
\end{proof}

Now, as a corollary to Lemma~\ref{lem:ilevelcoreset}, we get the following using
Lemma~\ref{lem:binomgeom}.

\begin{corollary}
  \label{cor:ilevelcoreset}
  If $C^\ell_w$ is an $\ell$-level $(\eps/(2\ell))$-coreset of $C_w$, then
  $C^\ell_w$ is an $\eps$-coreset of $C_w$.
\end{corollary}

As we discussed earlier, rounding of the weights at internal nodes is needed in
our dynamic algorithm to achieve the desired worst-case update time.  Towards
that, we need two lemmas.

In the next lemma, think of $a/b$ as the original weight of the point, $c/d$ as
the weight that we want to approximate $a/b$ with, and $D$ as the cost of this
point with respect to a feasible solution in $Q$.  So the lemma says that by
rounding, the cost of the point stays within $1 \pm b/d$ of the original cost.
\begin{lemma}
  \label{lem:round}
  For positive integers $a$, $b$, and $d$, let $c = \lfloor ad/b \rfloor$.  Then
  $cD/d \in [1 \pm b/d] aD/b$ for any nonnegative real $D$.
\end{lemma}
\begin{proof}
  By the definition of $c$, we have that $c/d \le a/b \le c/d + 1/d$, and
  $1/d \le a/d$ because $a \ge 1$; hence $a/b \ge c/d \ge a/b - a/d$, which
  implies that $aD/b \ge cD/d \ge aD/b - aD/d = (1 - b/d) aD/b$.
\end{proof}

The proof of the following lemma is very similar.  Here, think that we
approximate the weight $r$ of a point by $\lfloor r \rfloor + c/d$ and the cost
of the point stays within $1 \pm 1/d$ of the original cost.
\begin{lemma}
  \label{lem:roundr}
  Let $r \ge 1$ be a rational number, $a$ and $b$ be positive integers such that
  $a/b = r - \lfloor r \rfloor$, $d$ be any positive integer, and
  $c = \lfloor ad/b \rfloor$. Then
  $(\lfloor r \rfloor + c/d)D \in (1 \pm 1/d)rD$ for any nonnegative real $D$.
\end{lemma}
\begin{proof}
  By the definition of $c$ and using $r \ge 1$, we get that
  $a/b \ge c/d \ge a/b - r/d$; adding $\lfloor r \rfloor$ and multiplying by $D$
  finishes the proof.
\end{proof}


\section{A Dynamic Coreset}
We describe our dynamic algorithm for maintaining an $\eps$-coreset for a
problem $P$ with query space $Q$ that uses a static coreset algorithm, say,
$\alg_S$.
\begin{figure}[ht]
\centering
\begin{tikzpicture}
  [scale=0.75,minimum size=20pt,
  vert/.style={circle,draw,thick},
  alg/.style={rectangle,draw,thick}]
  \node (v1) at (2,0) [vert] {$v_1$};
  \node (v2) at (4,0) [vert] {$v_3$};
  \node (v3) at (6,0) [vert] {$v_2$};
  \node (v4) at (8,0) [vert] {$v_4$};
  \node (a1) at (3, 2) [alg] {$\alg^2_S$};
  \node (a2) at (7, 2) [alg] {$\alg^3_S$};

  \foreach \x in {1,...,2}
  \node (v2\x) at (-1 + 4*\x, 4) [vert] {$v^2_\x$};

  \node (ao) at (5, 6) [alg] {$\alg^1_S$};

  \node (o) at (5, 8) {Output};

  \draw[thick,->] (v1) -- (a1);
  \draw[thick,->] (v2) -- (a1);
  \draw[thick,->] (v3) -- (a2);
  \draw[thick,->] (v4) -- (a2);

  \draw[thick,->] (a1) -- (v21);
  \draw[thick,->] (a2) -- (v22);
  
  \draw[thick,->] (v21) -- (ao);
  \draw[thick,->] (v22) -- (ao);

  \draw[thick,->] (ao) -- (o);


\end{tikzpicture}
\caption{An $\alg_S$ node takes input from two point-nodes.  If the union of the
  sets has cardinality greater than $s'$, then the $\alg_S$ node computes a
  coreset of cardinality at most $s'$ and passes it on to the point-node above
  it (its parent).  The number of leaf nodes is always $n$, and the number of
  levels is always $O(\log n)$, where $n$ is the current number of points.}
\label{fig:dyncor}
\end{figure}
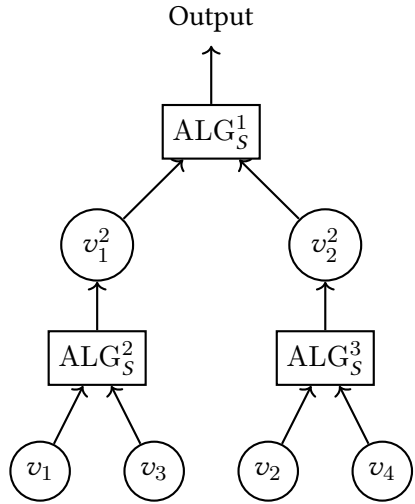


The main idea is described in Figure~\ref{fig:dyncor} using a tree with a
special structure. Each node is of one of the two types: a point-node
representing a weighted set of points or an alg-node representing an instance of
$\alg_S$.  We sometimes use a point-node to denote the point set it represents
and an alg-node to denote the $\alg_S$ instance it represents.  Each level
contains either only point-nodes or only alg-nodes.  All leaf nodes are
point-nodes and represent a weighted singleton with an input point.  Each
alg-node gets as input the weighted union of its children, and its output is
represented by its parent node (which is a point-node).  When running $\alg_S$
at an alg-node $A$, if the union of its children has cardinality larger than
$s'$, then $A$ would compute a coreset of cardinality at most $s'$ otherwise it
would just output the weighted union.  We will later fix this threshold $s'$ for
computing a coreset.  An example of how insertions and deletions are handled is
shown in Figure~\ref{fig:treeprog} (where all weights are assumed to be one).
For the ease of description, from now onwards, we will think of this tree with
alg-nodes being collapsed into their parent nodes.  Then each leaf node would
contain a weighted singleton and each internal node would contain the output of
the $\alg_S$ instance run on the weighted union of its children's sets.

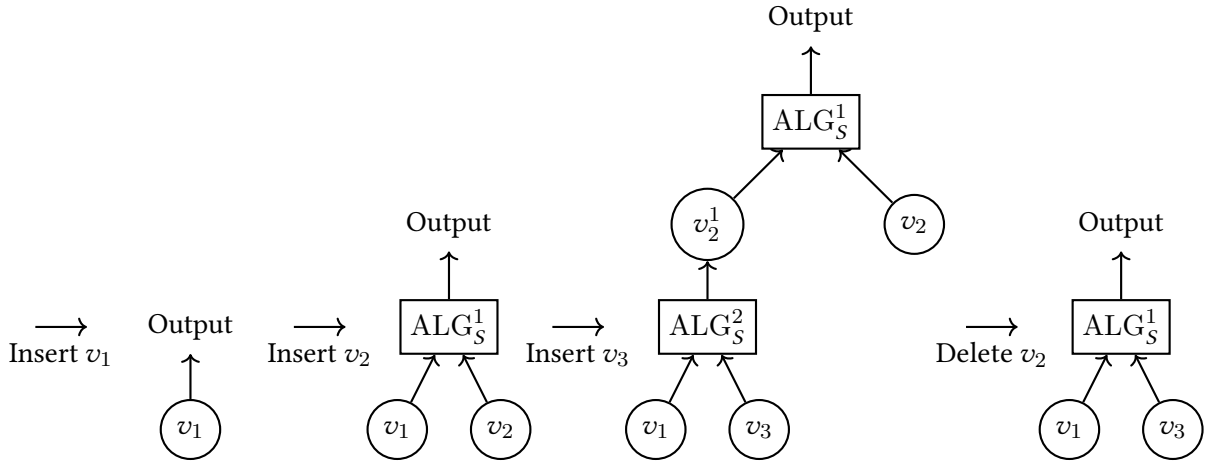
\begin{figure}[ht]
\centering
\begin{tikzpicture}
  [scale=0.68,minimum size=20pt,
  vert/.style={circle,draw,thick},
  alg/.style={rectangle,draw,thick}]
  \begin{scope}[xshift=-4cm]
    \draw[thick,->] (-1, 2) -- node[below]{Insert $v_1$} (0, 2);
    \node (v1) at (2, 0) [vert] {$v_1$};

    \node (o) at (2, 2) {Output};

    \draw[thick,->] (v1) -- (o);
  \end{scope}

  \begin{scope}
    \draw[thick,->] (0, 2) -- node[below]{Insert $v_2$} (1, 2);
    \node (v1) at (2, 0) [vert] {$v_1$};
    \node (v2) at (4, 0) [vert] {$v_2$};
    \node (a1) at (3, 2) [alg] {$\alg^1_S$};

    \node (o) at (3, 4) {Output};

    \draw[thick,->] (v1) -- (a1);
    \draw[thick,->] (v2) -- (a1);

    \draw[thick,->] (a1) -- (o);
  \end{scope}

  \begin{scope}[xshift=5cm]
    \draw[thick,->] (0, 2) -- node[below]{Insert $v_3$} (1, 2);
    \node (v1) at (2, 0) [vert] {$v_1$};
    \node (v2) at (4, 0) [vert] {$v_3$};
    \node (a1) at (3, 2) [alg] {$\alg^2_S$};

    \node (v12) at (3, 4) [vert] {$v^1_2$};
    \node (v3) at (7, 4) [vert] {$v_2$};

    \node (ao) at (5, 6) [alg] {$\alg^1_S$};

    \node (o) at (5, 8) {Output};

    \draw[thick,->] (v1) -- (a1);
    \draw[thick,->] (v2) -- (a1);

    \draw[thick,->] (a1) -- (v12);

    \draw[thick,->] (v12) -- (ao);
    \draw[thick,->] (v3) -- (ao);

    \draw[thick,->] (ao) -- (o);
  \end{scope}
  
  \begin{scope}[xshift=13cm]
    \draw[thick,->] (0, 2) -- node[below]{Delete $v_2$} (1, 2);
    \node (v1) at (2, 0) [vert] {$v_1$};
    \node (v2) at (4, 0) [vert] {$v_3$};
    \node (a1) at (3, 2) [alg] {$\alg^1_S$};

    \node (o) at (3, 4) {Output};

    \draw[thick,->] (v1) -- (a1);
    \draw[thick,->] (v2) -- (a1);

    \draw[thick,->] (a1) -- (o);
  \end{scope}

\end{tikzpicture}
\caption{An example of how insertions and deletions are handled.  We start with
  an empty tree.  The first point that is inserted is represented by $v_1$.  We
  use a point and the node that represents it interchangeably.  Then $v_2$ is
  inserted followed by $v_3$.  Next, if $v_4$ is inserted, we get exactly the
  tree shown in Figure~\ref{fig:dyncor}, and if $v_2$ is deleted, then we get
  the last tree.}
\label{fig:treeprog}
\end{figure}


We guarantee that the resulting tree then will always be a \emph{complete}
binary tree, i.e., every level except possibly the lowest is completely filled,
and the nodes at the lowest level are packed to the left.  To describe the
updates briefly, let $\ell_r$ denote the rightmost leaf node at the lowest
level; for simplicity, assume that the lowest level is not full.  Insertion is
straightforward: the new point goes in a new leaf node to the right of $\ell_r$.
For deletion of a point at leaf node $\ell_d$, if $\ell_d \neq \ell_r$, then we
replace contents of $\ell_d$ with those of $\ell_r$ and delete $\ell_r$.  See
Section~\ref{sec:heap_description} for details of these operations.  For weight
update, the tree does not change.

\begin{remark}
  Since a coreset will not be computed until a node has more than $s'$ points,
  the tree can be modified so that each leaf node corresponds to a set of
  $\Theta(s')$ points.  Then the number of nodes in the tree is $\Theta(n/s')$.
  This reduces the additional space used for maintaining this tree.  This is
  important when the number of points is very large.  See
  Section~\ref{sec:treesizereduction} for further details.  This is essentially
  the same idea as used for asymmetric sparsification in Section~3.4 in Eppstein
  et al.~\cite{eppstein_etal97}.
\end{remark}

We call the leaf nodes at the same level as that of the leftmost leaf node to be
at level~$0$.  We increment these level numbers naturally as we move upwards in
the tree.  Since we maintain a complete binary tree, the root, which is at the
highest level, is on level~$\clog{n}$.  After a point insertion, deletion, or
weight update, we recompute all the nodes that are affected by running $\alg_S$
from scratch.  Once we update a leaf node, all the nodes on its leaf-to-root
path are affected.
Since at most two leaf nodes are updated after every point update, we run at
most $2\clog{n}$ instances of $\alg_S$.  Finally, to reduce the cardinality of
our output coreset, we run another \emph{outer} instance of $\alg_S$ with
$\eps_s = \eps/3$ and $\errpr_s = \errpr/2$ with input as the output of the
root.  Here, $\eps_s$ and $\errpr_s$ are parameters for $\alg_S$ as described
below, and our goal is to compute an $\eps$-coreset with probability at least
$1{-} \errpr$.  The outer instance is run after every update.

The static coreset algorithm $\alg_S$ takes as input an integer weighted set of
$n_s$ points with total weight $W_s$ and always returns a weighted set of
cardinality at most $s(\eps_s, \errpr_s, W_s)$; this set is an $\eps_s$-coreset
with probability at least $1{-} \errpr_s$.  Let the running time of $\alg_S$ be
$t(n_s, \eps_s, \errpr_s, W_s)$.  We assume that the functions $t$ and $s$ are
nondecreasing in $W_s$ and nonincreasing in $\eps_s$ and $\errpr_s$, and also
that $t$ is nondecreasing in $n_s$. We call such functions $t$ and $s$
\emph{well-behaved.}

We note that $t$ and $s$ implicitly depend on the query space $Q$ as well.  In
particular, for $k$-median and $k$-means, they depend on $k$ and the dimension
or the cardinality of the universe from which a solution is allowed to be
picked.  Also, assume that the total weight of $\alg_S$'s output is at most
$1 {+} \delta$ times the total input weight and it outputs a coreset of points
with integer weights.  For the dynamic algorithm, $n$ denotes the current number
of points, and we assume that any input weight is a rational number with
numerator and denominator bounded by $n^c$, for a fixed constant $c$.
\begin{theorem}
  \label{thm:main}
  Assume that there is a static algorithm $\alg_S$ that takes as input an
  integer-weighted set of $n_s$ points with total weight $W_s$ and always
  returns an integer-weighted set of cardinality at most
  $s(\eps_s, \errpr_s, W_s)$ with total weight at most $(1 {+} \delta)W_s$, and
  this set is an $\eps_s$-coreset with probability at least $1{-} \errpr_s$.
  Let the running time of $\alg_S$ be $t(n_s, \eps_s, \errpr_s, W_s)$, and
  assume that both $s$ and $t$ are well-behaved.  Then there is a fully-dynamic
  algorithm that, on \emph{rational}-weighted input points, always maintains an
  $s\left(\frac{\eps}{3}, \frac{\errpr}{2}, W_p\right)$-cardinality weighted set.
  This set is an $\eps$-coreset with probability at least $1{-} \errpr$.  Its
  worst-case update time is
  \[
    O\left(t\left(2s^*, \frac{\eps}{6\clog{n_p}}, \frac{\errpr}{2 n_p}, W_p\right)
      \cdot \left(1 + \log (1 {+} \delta) + \frac{\log \eps^{-1}}{\log n}\right)
      \cdot \log n\right)\,,
  \]
  where $8n/3 \le n_p \le 8n$,
  $W_p = (1 {+} \delta)^{\clog{n_p}}n_p^{c''}\ceile{1/\eps}$, $c''$ is a
  constant, and
  $s^* = s\left(\frac{\eps}{6\clog{2 n_p}}, \frac{\errpr}{4 n_p}, W_p\right)$.
\end{theorem}
\begin{proof}
  We first prove that the output of the algorithm is an $\eps$-coreset if every
  non-outer $\alg_S$ instance outputs an $\eps_s$-coreset of its input for some
  $\eps_s \le \eps/(6\clog{n})$ and the outer $\alg_S$ instance outputs an
  $(\eps/3)$-coreset of its input.  We prove the following by induction on level
  number: every node at level~$\ell$ contains a
  $(\sum_{i = 1}^\ell\binom{\ell}{i}\eps_s^i)$-coreset of the leaf nodes in its
  subtree.  In the base case, a node at level~$1$ contains an $\eps_s$-coreset
  of its input trivially.  An $\alg_S$ instance $A$ at level~$i$ gets as input
  two sets, say $C'_w$ and $C''_w$, each of which is a
  $(\sum_{i = 1}^{\ell - 1}\binom{\ell - 1}{i}\eps_s^i)$-coreset for the leaf
  nodes in their respective nodes' subtrees.  Hence, $C'_w \cup C''_w$ is a
  $(\sum_{i = 1}^{\ell - 1}\binom{\ell - 1}{i}\eps_s^i)$-coreset for leaf nodes
  in the subtree rooted at $A$ by Lemma~\ref{lem:coresetunion}.  Now, $A$
  outputs an $\eps_s$-coreset of $C'_w \cup C''_w$, hence by
  Lemma~\ref{lem:coresetcomposition}, its output is an
  $(\eps_s + (1{+} \eps_s)\sum_{i = 1}^{\ell - 1}\binom{\ell -
    1}{i}\eps_s^i)$-coreset of the leaf nodes in its subtree, which, by
  Lemma~\ref{lem:binom}, means a
  $(\sum_{i = 1}^\ell\binom{\ell}{i}\eps_s^i)$-coreset.  This completes the
  induction step.  Hence, the root node, which is at level~$\clog{n}$, contains
  $(\sum_{i = 1}^{\clog{n}}\binom{\clog{n}}{i}\eps_s^i)$-coreset.  Now, since
  $\eps_s \le \eps/(6\clog{n})$, by Lemma~\ref{lem:binomgeom}, the output at the
  root is an $(\eps/3)$-coreset.  The outer $\alg_S$ instance outputs an
  $(\eps/3)$-coreset of this, hence, by Lemma~\ref{lem:coresetcomposition}, the
  final output is an $(2\eps/3 + \eps^2/9)$-coreset, which is an $\eps$-coreset
  of all points.

  Recall that the running time of $\alg_S$ is $t(n_s, \eps_s, \errpr_s, W_s)$ to
  compute an $\eps_s$-coreset with probability at least $1{-} \errpr_s$, where
  $n_s$ is the number of points in the input.  Our output success probability
  will depend on $\errpr_s$, and $\eps$ depends on $\eps_s$ as proved in the
  previous paragraph.  We will need $\eps_s \le \eps/(6\clog{n})$ and
  $\errpr_s \le \errpr/(2n)$, so these depend on $n$, which can change a lot
  over time.  We now show how to maintain these guarantees for $\eps_s$ and
  $\errpr_s$ after each update.
  
  Towards this, we need a little tweak to our algorithm and an additional
  maintenance routine that we call the \emph{refresher}.  The algorithm works in
  phases.  The refresher routine maintains two \emph{refresh} pointers that
  always point to consecutive leaf nodes, say $r_1$ and $r_2$.  The refresh
  pointers are reset after the end of a phase as follows.  If the number of leaf
  nodes is a power of $2$, then $r_1$ and $r_2$ point to the two leftmost leaf
  nodes, otherwise they point to the two leftmost leaf nodes at the level above
  the lowest level.  Assume, for completeness, that the very first phase ends
  after receiving two points, so the tree is just two leaf nodes and their
  parent as the root.

  For each subsequent phase, let $n_0$ be the value of $n$ at the beginning of
  the phase.  Each phase ends after $n_0/2$ updates, and we set $n_p = 4 n_0$.
  This guarantees that $n_p$ is greater than $n$ throughout the whole phase and
  even the next phase (details appear below).  After receiving an update, we
  rerun all the $\alg_S$ instances on the leaf-to-root path starting at $r_1$
  and $r_2$ (at most $2\clog{n}$ such instances).  This is the refresher
  routine.  Then we move the refresh pointers to the next two leaf nodes on the
  right.  If we reach the right end, then we go to the next level if it exists,
  otherwise we stop.  If we stop, then we achieved the goal of (re-)running all
  the $\alg_S$ instances that are present at the end of the phase at least once
  in this phase (this will become clearer below).  After the refresher routine,
  we execute the update which affects at most two leaf nodes. We rerun all the
  $\alg_S$ instances that are affected by this update, again, at most
  $2\clog{n}$ such instances.  So in total, at most $4\clog{n}$ of non-outer
  $\alg_S$ instances are run after an update and one outer instance, which
  explains the $\log n$ factor in the update time.  We now explain the
  parameters used in the $\alg_S$ instances.  For all the non-outer $\alg_S$
  instances, we use $\eps_s = \eps/(6\clog{n_p})$ and
  $\errpr_s = \errpr/(2 n_p)$.  (This explains the $\eps_s$ and $\errpr_s$
  parameters of the functions $t$ and $s$ in the theorem statement.)  Note here
  that the running time of the outer instance is going to be less than any
  non-outer instance because $t$ is non-increasing in $\eps_s$ and $\errpr_s$.
 
  As we use $n_p = 4 n_0$ and there could be at most $n_0/2$ insertions in a
  phase, the final value of $n$ is at most $3 n_0/2$, and, thus, $n_p$ is always
  greater than $n$.  In fact, crucially, $n_p$ is an upper bound on $n$ for even
  the next phase; in the next phase,
  $n \le n_0 + n_0/2 + (n_0 + n_0/2)/2 = 9 n_0/4 \le n_p$.  Also, in the current
  phase, $n_0/2 \le n \le 3n_0/2$, hence $8n/3 \le n_p \le 8n$, as required
  (cf.~the theorem statement).

  We now prove that any non-outer $\alg_S$ instance uses
  $\eps_s\le \eps/(6\clog{n})$ and $\errpr_s \le \errpr/(2n)$ at any time
  instant.  Let $L$ be the set of leaf nodes at the beginning of the phase;
  therefore, $|L| = n_0$.  An $\alg_S$ instance that exists at the end of the
  phase is either on the leaf-to-root path for some leaf in $L$ or it was
  created/updated in this phase.  At the end of the phase, the refresh pointers
  will hit all surviving leaf nodes in $L$; the argument is as follows.  Each
  phase lasts for $n_0/2$ updates, $|L| = n_0$, and we move the two refresh
  pointers to the right on next two leaf nodes after each update.  Importantly,
  new leaf nodes are added only to the right of the rightmost leaf node at the
  lowest level, and hence, the refresher routine will have hit all surviving
  leaf nodes in $L$ before hitting a newly created leaf node.

  This shows that, in any case (being either hit by a point update or by the
  refresher routine), each $\alg_S$ instance is run with $n_p = 4 n_0$, setting
  up these instances for the next phase.  This means that at any time instant,
  each $\alg_S$ instance was created/updated in the current phase or
  created/updated in the previous phase, thus showing that
  $\eps_s\le \eps/(6\clog{n})$ and $\errpr_s \le \errpr/(2n)$ for all $\alg_S$
  instances at all times.

  At any time instant, there are at most $n$ non-outer instances of $\alg_S$,
  each with success probability at least $1 - \errpr/(2n)$, and the outer
  $\alg_S$ instance has success probability at least $1 -\errpr/2$.  Hence, the
  final success probability is at least $1{-} \errpr$ by the union bound over
  these $n + 1$ instances.

  \paragraph{How to handle weights} We will need one further tweak to argue that
  each weight ever encountered by the algorithm can be stored using
  $O(1 + \log (1 {+} \delta) + \log(1/\eps)/\log n)$ words, which also explains
  that factor in the update time.  By assumption, an insertion or weight update
  comes with a weight that is a fraction with the numerator and the denominator
  bounded by $n^c$ for some fixed constant $c$.  After receiving such an update,
  we approximate the weight by a fraction that has numerator bounded by
  $n_p^{c'}\ceile{1/\eps}$, where $c' = 2c + 1$ is also a fixed constant, and
  the denominator is equal to $n_p^{c + 1}\ceile{1/\eps}$\footnote{The static
    algorithm $\alg_S$ expects integer-weighted input and outputs
    integer-weighted points, whereas our dynamic algorithm handles fractional
    weights.  If fractional weights are na\"{i}vely stored in our dynamic
    algorithm, then at internal nodes, combining two fractions may result in
    larger magnitude numbers.  E.g., na\"{i}vely handling two points with
    weights $a/b$ and $c/d$ so as to be used in $\alg_S$ results in weights
    $ad/(bd)$ and $bc/(bd)$.  Thus, at level~$i$, the numerators and
    denominators may be as large $(\poly(n))^{2^i}$.  Note that some rounding
    would be needed even if $\alg_S$ can handle rational weights, because its
    output may be points with rational weights having much larger magnitude;
    e.g., even if the output magnitude is about only quadratic in that of the
    input, the blowup near the root in our dynamic algorithm would be $n$th
    power of the input.  In fact, we do this rounding in the proof of
    Theorem~\ref{thm:kmedkmeans}.}.  The change in the cost due to this
  approximation is at most $\eps /n_p$ times the original cost; hence, by the
  linearity of the cost function, the output coreset quality is affected by at
  most an additive factor of $O(\eps/n)$.  More formally, the following claim
  holds by Lemma~\ref{lem:round} and using $b/d \le \eps/n_p$ below (think of
  $D$ below as cost).
  \begin{claim}
    \label{clm:12}
    Let $d = n_p^{c + 1}\ceile{1/\eps}$.  Given a rational number $a/b$, where
    $a$ and $b$ are integers, $a \le n_p^c$ and $b \le n_p^c,$ let
    $f = \lceil ad/b \rceil$. Then $f \le n_p^{2c + 1}\ceile{1/\eps}$ and
    $(f/d)D \in [1 {\pm} \eps/n_p] (a/b)D$ for any nonnegative real $D$.
  \end{claim}
  Recall that due to the refresher routine, at any time instant, the denominator
  of the weight at any leaf node can be one of the two:
  $n_p^{c + 1}\ceile{1/\eps}$ or $n_{pp}^{c + 1}\ceile{1/\eps}$, where $n_{pp}$
  is the value of $n_p$ for the previous phase. When the two children of an
  internal node use different denominators, this complicates our rounding
  scheme.  Thus, when taking a union of the children's sets at an internal node,
  for each weight, we make its numerator an integer and the denominator equal to
  $(n_p n_{pp})^{c + 1}\ceile{1/\eps}$, which is a common multiple of
  $n_p^{c + 1}\ceile{1/\eps}$ and $n_{pp}^{c + 1}\ceile{1/\eps}$---the only
  possible denominators of an input weight after rounding.  Next, we run the
  $\alg_S$ instance with integer weights as given by the numerator, then
  (implicitly) dividing the output weights by the denominator
  $(n_p n_{pp})^{c + 1}\ceile{1/\eps}$ afterwards.  Since each $\alg_S$ instance
  can increase the total weight by at most a factor of $1 {+} \delta$, the sum
  of the numerators of all weights at level~$i$ is always bounded by
  $n(1 {+} \delta)^i (n_p n_{pp})^{c'}\ceile{1/\eps}$.  Since $i \le \clog{n}$
  and $n_{pp} = \Theta(n_p)$, there exists a constant $c''$, such that the sum
  of the numerators of all weights at any level~$i$ and all the possible
  numerators and denominators are bounded by
  $(1 {+} \delta)^{\clog{n_p}}n_p^{c''}\ceile{1/\eps} =: W_p$, and hence, can be
  stored in $O(1 + \log (1 {+} \delta) + \log(1/\eps)/\log n)$ words as desired
  (see the beginning of the paragraph before Claim~\ref{clm:12}).  This also
  justifies the $W_s$ parameters of the functions $t$ and $s$ in the theorem
  statement.

  Now we put everything together.  The outer $\alg_S$ instance outputs a
  weighted set of size at most
  $s\left(\frac{\eps}{3}, \frac{\errpr}{2}, W_p\right)$.  This set is an
  $\eps$-coreset with probability at least $1{-} \errpr$, which we proved by a
  union bound over all $\alg_S$ instances.  We set
  $s' = s\left(\frac{\eps}{6\clog{n_p}}, \frac{\errpr}{2 n_p}, W_p\right)$,
  which is the threshold for computing a coreset at each internal node, i.e.,
  (recall that) if the number of points at an internal node is greater than
  $s'$, then we run $\alg_S$ to compute a coreset.  An upper bound on the
  threshold for the current phase and the previous phase is
  $s^* = s\left(\frac{\eps}{6\clog{2 n_p}}, \frac{\errpr}{4 n_p}, W_p\right)$
  because the $n_p$ value for the previous phase can be at most twice that of
  the current phase.  Then the worst-case update time is dominated by the
  non-outer $\alg_S$ instances, each running in time
  $t\left(2 s^*, \frac{\eps}{2\clog{n_p}}, \frac{\errpr}{2 n_p}, W_p\right)$,
  and we run $O(\log n)$ of these after receiving an update.  An additional
  factor of $1 + \log(1 {+} \delta) + \log(1/\eps)/\log n$ appears because each
  weight may need memory worth $O(1 + \log(1 {+} \delta) + \log(1/\eps)/\log n)$
  words, and we need constant time to access each memory word.
\end{proof}

Before proving the concrete bounds for $k$-median and $k$-means that are stated
in Theorem~\ref{thm:kmedkmeans}, we prove a weaker theorem that is a direct
consequence of Theorem~\ref{thm:main} using the static algorithm of
Chen~\cite{chen_09}.
\begin{theorem}
  \label{thm:kmedkmeansweak}
  For the $k$-median and $k$-means problems, there is a fully-dynamic algorithm
  that maintains a set of cardinality
  $O(\eps^{-2}k\log^2 (n/\eps)(k\log n + \log(1/\errpr)))$, that is an
  $\eps$-coreset with probability at least $1{-} \errpr$, and has worst-case
  update time
  \[
    O\left(\eps^{-2}k^2\log^3 n \log^2
      \frac{n}{\eps}\log\frac{n}{\errpr}\left(k\log n +
        \log\frac{n}{\errpr}\right)\log\log\frac{n}{\eps}\left(1+\frac{\log
          \eps^{-1}}{\log n}\right)\right)\,.
  \]

  Ignoring the $\log \log n$ factors, for $\errpr = \Omega(1/\poly(n))$ and
  $\eps = \Omega(1/\poly(n))$, the coreset cardinality is
  $O(\eps^{-2}k^2\log^3 n)$, and the worst-case update time is
  $O(\eps^{-2}k^3\log^7 n)$.
\end{theorem}
\begin{proof}
  Chen's algorithm takes in an integer weighted set and outputs also an integer
  weighted set.  Its output has the same total weight as the input, so
  $\delta = 0$ (see Theorem~\ref{thm:main}).  Also, for Chen's algorithm,
  $s(\eps_s, \errpr_s, W_s) = O(\eps_s^{-2}k(k\log n + \log(1/\errpr_s))\log^2
  W_s)$ and the running time
  $t(n_s, \eps_s, \errpr_s, W_s) = O(n_s k\log(1/\errpr_s) \log\log W_s)$ (see
  Theorems~3.6~and~5.5 in Chen~\cite{chen_09}), which is dominated by the
  computation of a bicriteria approximation. Note that both $s$ and $t$ are
  well-behaved. Using $W_p = O(\poly(n)/\eps)$,
  $s^* = O(\eps^{-2}k\log^2 n\log^2(n/\eps)(k\log n + \log(n/\errpr)))$, and
  $\delta = 0$ in Theorem~\ref{thm:main} gives the desired bounds using the
  functions $t$ and $s$ above.
\end{proof}

Now we use the result of Braverman et al.~\cite{braverman_etal_16} to get better
bounds as stated in Theorem~\ref{thm:kmedkmeans} in the introduction section.
Unfortunately, we cannot use Theorem~\ref{thm:main} as a complete black box for
this because in this case, on integer weighted input, $\alg_S$ does not output
an integer weighted coreset.  The proof of the following theorem is thus an
extension of the proof of Theorem~\ref{thm:main}.
\begingroup
\def\thetheorem{\ref{thm:kmedkmeans}}
\begin{theorem}
  For the $k$-median and $k$-means problems, there is a fully-dynamic algorithm
  that maintains a set of cardinality
  $O(\eps^{-2}k(\log n \log k \log(k\eps^{-1}\log n) + \log (1/\errpr)))$, that
  is an $\eps$-coreset with probability at least $1{-} \errpr$, and has
  worst-case update time
  $O\left(\eps^{-2} k^2\log^5n\log^3 k \log^2(1/\eps) (\log\log n)^3\right)$,
  assuming that $\eps = \Omega(1/\poly(n))$ and $\errpr = \Omega(1/\poly(n))$.
\end{theorem}
\addtocounter{theorem}{-1}
\endgroup
\begin{proof}
  Our dynamic algorithm expects to have at its disposal a static algorithm
  $\alg_S$ that takes integer-weighted input and outputs an integer-weighted
  coreset.  Since the algorithm of Braverman et al.\ that we use as $\alg_S$
  outputs on \emph{integer} weighted input a coreset with \emph{fractional}
  weights, we need some modifications.  Hence, before $\alg_S$ is ready to be
  used in the dynamic algorithm, we round its output to turn it into integers.
  \paragraph{Weight-Rounding Modifications for $\alg_S$}
  \mbox{}
  \begin{itemize}
    \item Let the input to $\alg_S$ be $Y_w$ which is a set of $n_s$ points with
    integer weights $w(1), \ldots, w(n_s)$.

    \item We scale these weights first.  We run $\alg_S$ on the same points with
    weights $s' w(1), \ldots, s' w(n_s)$, where $s'$ is the desired cardinality
    of the output coreset (which is the same as the threshold for computing a
    coreset at an internal node in this case).  We set $s'$ later in a such a
    way that it can be computed by our dynamic algorithm.  This step of
    multiplying input weights by $s'$ is done to make sure that each of the
    fractional weights output by $\alg_S$ is at least $1$ (see Line~6 of
    Algorithm~2 in Braverman et al.~\cite{braverman_etal_16}).
    
    \item Let the output $C_w$ of $ \alg_S$ be a weighted set of $s'$ points
    with fractional weights $w_o(1), \ldots, w_o(s')$.  Using the rounding
    strategy of Lemma~\ref{lem:roundr}, round these fractional weights to have
    an integer numerator and the denominator equal to $\ceile{(\log n_p)/\eps}$
    to get weights $\tilde{w}(1), \ldots, \tilde{w}(s')$, where $n_p$ is as
    defined in the proof of Theorem~\ref{thm:main}.  Formally, for
    $i \in \{1, \ldots, s'\}$:
    \[
      \tilde{w}(i) = \floore{w_o(i)} + \frac{\floorelr{\left(w_o(i) -
            \floore{w_o(i)}\right)\ceilelr{\frac{\log n_p}{\eps}}}}
      {\ceilelr{\frac{\log n_p}{\eps}}}\,.
    \]
    Since $w_o(i) \ge 1$, by Lemma~\ref{lem:roundr}, for any real $D \ge 0$, we
    have $\tilde{w}(i) D \in [1 \pm \eps/\log n_p] w_o(i)D$.
    
    \item Hence, by the linearity of the cost function, $C_w$ with weights
    $\tilde{w}(1)/s', \ldots, \tilde{w}(s')/s'$ is an
    $(\eps_s + 2\eps/\log n_p)$-coreset of $Y_w$ with weights
    $w(1), \ldots, \allowbreak w(n_s)$ if $C_w$ with weights
    $w_o(1), \ldots, w_o(s')$ is an $\eps_s$-coreset of $Y_w$ with weights
    $s' w(1), \ldots, \allowbreak s' w(n_s)$.  Note that $\tilde{w}(i)/s'$ can
    be represented as a fraction with an integer numerator and denominator equal
    to $s'\ceile{(\log n_p)/\eps}$.
    
    \item The additive loss of $2\eps/\log n_p$ in the coreset quality due to
    this rounding is tolerable because every non-outer $\alg_S$ instance will be
    run with $\eps_s = O(\eps/\log n_p)$\footnote{If we go for smaller additive
      loss, say $\eps/n_p$, the denominators of resulting numbers due to this
      rounding would become exponential in $n_p$.  And if we go for a larger
      additive loss, it would worsen the coreset quality at non-outer instances
      to $\omega(\eps/\log n_p)$ resulting in the quality of the output coreset
      worse than $\eps$.}.  Hence, the coreset quality at internal nodes will
    always be $O(\eps_s + \eps/\log n_p) = O(\eps/\log n)$, as desired.
    \item This rounding ensures that on integer-weighted input with total weight
    $W$, the output weights of $\alg_S$ are fractions with integer numerator
    bounded by $(1 {+} \delta) W s' \ceile{(\log n_p)/\eps}$ and integer
    denominator equal to $s' \ceile{(\log n_p)/\eps}$.  Here, $1{+} \delta$ is
    the factor by which $\alg_S$ can increase the total weight.
  \end{itemize}
  
  To handle rational weights in the dynamic algorithm, we first proceed as
  described in the paragraph on how to handle weights in the proof of
  Theorem~\ref{thm:main}.  Recall that we assume that each insertion or weight
  update by the adversary comes with a weight that is a fraction with the
  numerator and the denominator bounded by $n^c$ for some fixed constant $c$,
  and we set $c' = 2c+1$.  Also, each leaf node was created/updated in the
  current phase or created/updated in the previous phase and thus uses the value
  either $n_p$ or $n_{pp}$, where $n_{pp}$ is the value of $n_p$ for the
  previous phase.  We then showed the following.  At any time instant, the
  weight of the point at a leaf node is rounded in such a way that the numerator
  is bounded by $n_p^{c'}\ceile{1/\eps}$ and the denominator is equal to
  $n_p^{c + 1}\ceile{1/\eps}$, or the numerator is bounded by
  $n_{pp}^{c'}\ceile{1/\eps}$ and the denominator is equal to
  $n_{pp}^{c + 1}\ceile{1/\eps}$.  Due to this rounding, the output coreset
  quality is affected by at most an additive factor of
  $\max\{2\eps/n_p, 2\eps /n_{pp}\} = O(\eps/n)$.  We now prove the following
  more general statement towards the current proof.
  \begin{lemma}
    \label{lem:denom}
    At any time instant, every weight at a node at level~$i$ has an integer
    numerator and a denominator that is a factor of
    $(n_p n_{pp})^{c + 1}\ceile{1/\eps}(s'_p s'_{pp}\ceile{(\log
      n_p)/\eps}\ceile{(\log n_{pp})/\eps})^i =: D(i)$, where $s'_p$ and
    $s'_{pp}$ are values of the threshold $s'$ in the current and the previous
    phase, respectively.
  \end{lemma}
  \begin{proof}
    We prove this statement by induction over the sequence of nodes updated by
    the algorithm.

    In the base case, the first ever node update will be due to creation of a
    leaf node, and the weight will have denominator $n_p^{c+1}\ceile{1/\eps}$.
    Next we discuss the induction step.  Let the update be on a node at level
    $i$, so we run the \emph{modified} $\alg_S$ instance with all weights having
    a denominator that is a factor of $D(i{-}1)$, which is true by induction
    hypothesis.  Then, since the modified $\alg_S$ adds a factor of
    $s'_p\ceile{(\log n_p)/\eps}$ to the denominator, all resulting output
    weights have a denominator that is a factor of
    $D(i{-}1)s'_p\ceile{(\log n_p)/\eps}$, which is a factor of $D(i)$.  This
    finishes the induction step for the case when the node update is not the
    last of the phase.  When the node being updated is the last of the phase, we
    have to be careful.  In this case, we need to show that for all weights in
    all nodes, $n_{pp}$ or $s'_{pp}$ do not appear in the denominator, as this
    will set these denominators for the next phase.  Towards this, we need the
    following claim.
    \begin{claim}
      \label{clm:denom}
      Let $u$ be a node at level~$i$.  Fix a time instant.  Suppose, in the
      current phase, all nodes in the subtree rooted at $u$ were updated and $u$
      was updated after the update of the last-updated leaf node in the subtree.
      Then the denominator of the weights at $u$ is a factor of
      $n_p^{c + 1}\ceile{1/\eps}(s'_p \ceile{(\log n_p)/\eps})^i$ at the fixed
      time instant.
    \end{claim}
    We omit the proof of this claim as it can be proved easily by induction on
    the level number at any fixed time instant.

    After the last node update of the phase, every node in the tree has been
    updated in the current phase and the premise of Claim~\ref{clm:denom} holds
    due to the refresher routine.  Hence, by Claim~\ref{clm:denom}, after the
    last node update of the phase, i.e., just before the new phase begins, all
    denominators at level~$i$ are a factor of
    $n_p^{c + 1}\ceile{1/\eps}(s'_p \ceile{(\log n_p)/\eps})^i$.  Since $n_p$
    and $s'_p$ of this phase will become $n_{pp}$ and $s'_{pp}$ in the next
    phase, the induction hypothesis stays true for the next phase as well.  This
    finishes the proof of Lemma~\ref{lem:denom}.
  \end{proof}
  
  Since an $\alg_S$ instance may increase the total weight by at most a factor
  of $1 {+} \delta$, the sum of the numerators of weights at any level~$i$
  is at most
  $n_p (1 {+} \delta)^i (n_p n_{pp})^{c'}\ceile{1/\eps}(s'_p s'_{pp}\ceile{(\log
    n_p)/\eps}\ceile{(\log n_{pp})/\eps})^i$; this can be seen by an easy
  induction on the level number.  Using this bound, we set the threshold $s'$ in
  a way similar to that in the proof of Theorem~\ref{thm:main}: we set
  $s'_p = s(\eps/(6\clog{n_p}), \errpr/(2 n_p), W_p)$, where
  \[
    W_p = (1 {+} \delta)^{\clog{n_p}}n_p^{c_1}\left(k\ceilelr{\frac{\log
          n_p}{\eps}}\right)^{c_2\clog{n_p}}\,,
  \]
  and $c_1$ and $c_2$ are chosen to be large enough constants so that $W_p$
  upper bounds the sum of the numerators of all weights at any level.  From now
  onwards, we assume that $\errpr = \Omega(1/\poly(n))$.  For $\alg_S$, the
  function $s$ is
  $s(\eps_s, \errpr_s, W_s) = O(\eps_s^{-2}k(\log k \log W_s +
  \log(1/\errpr_s)))$ and $\delta = O(\eps)$.  Then, using
  $n_{pp} = \Theta(n_p)$, we get that both $s'_p$ and $s'_{pp}$ are
  $O\left(\left(k\ceilelr{\frac{\log n_p}{\eps}}\right)^{c_3}\right)$, where
  $c_3$ is a fixed constant (so, independent of $c_1$ and $c_2$).  Observe that
  $W_p$ and thus $s'_p$ are determined by the phase and hence can be computed by
  our algorithm.  More concretely, we get that both $s'_p$ and $s'_{pp}$ are
  \[
    O\left(\eps^{-2}k\log^3 n \log k \log\frac{k\log n}{\eps}\right)\,.
  \]
  All possible numerators and denominators encountered by the algorithm are
  bounded by
  \[
    N := O\left(\poly(n) \left(\frac{k \log n }{\eps}\right)^{O(\log
        n)}\right)\,,
  \]
  so, can be stored in $m := (\log N)/\log n = O(\log ((k\log n)/\eps))$ words.

  The running time of $\alg_S$ is
  $t(n_s, \eps_s, \errpr_s, W_s) = O(n_s k\log(1/\errpr_s) \log\log W_s)$,
  which, similar to Chen's algorithm, is dominated by computation of a
  bicriteria approximation.  At a non-outer $\alg_S$ instance, $n_s = O(s'_p)$,
  $\eps_s = O(\eps/\log n_p)$, $\errpr_s = O(\errpr/n_p)$, and $W_s \le W_p$.
  With every update, $O(\log n)$ instances of $\alg_S$ are run, and an
  additional $m$ factor appears because a weight may need up to $m$ words.
  Hence, the worst-case update time assuming $\eps = \Omega(1/\poly(n))$ and
  $\errpr = \Omega(1/\poly(n))$ is
  \[
    O\left(t\left(s'_p, \frac{\eps}{\log n}, \frac{\errpr}{n}, W_p\right) m \log
      n\right) = O\left(\eps^{-2} k^2\log^5n\log k \log^2\left(\frac{k \log n
        }{\eps}\right) \log\log \left(\frac{k \log n }{\eps}\right)\right)\,,
  \]
  and a looser, easier to parse, bound is
  $O\left(\eps^{-2} k^2\log^5n\log^3 k \log^2(1/\eps) (\log\log n)^3\right)$.
  The output coreset cardinality is
  \[
    s\left(\frac{\eps}{3}, \frac{\errpr}{2}, W_p\right) = O\left(\eps^{-2}
      k\left(\log n \log k \log\left(\frac{k \log n }{\eps}\right) + \log
        \frac{1}{\errpr}\right)\right)\,.
  \]
This finishes the proof of Theorem~\ref{thm:kmedkmeans}.
\end{proof}


\subsection{The Binary-Tree Structure}
\label{sec:heap_description}
We describe the tree structure in more detail, especially, how insertions and
deletions are handled.  We always maintain a complete binary tree, in which
every level except possibly the lowest is completely filled, and the nodes in
the lowest level are packed to the left.  We also maintain the property that
each internal node has exactly two children.  Our data structure behaves
somewhat like a heap, though a crucial difference is that we do not have keys.
This structure supports insertion and deletion of a leaf node.  Insertion of a
new leaf-node $\ell$ works as follows.
\begin{itemize}
  \item If the current number of leaf nodes is a power of $2$, then let $v$ be
  the leftmost leaf node,
  \item Else let $v$ be the leftmost leaf node in the level above the lowest
  level.
  \item Let $p$ be $v$'s parent.
  \item Create a new node $u$.
  \item Make $p$ to be $u$'s parent; $u$ replaces $v$, so if $v$ was $p$'s right
  (respectively, left) child, then $u$ is now $p$'s right (respectively, left)
  child.
  \item Make $v$ to be $u$'s left child and $\ell$ to be $u$'s right child.
  This way, $\ell$ the rightmost leaf node at the lowest level.
\end{itemize}
Deletion of a leaf-node $\ell$ works as follows.  Let $v$ be the rightmost leaf
node at the lowest level, $p$ be $v$'s parent, and $v'$ be $v$'s sibling.
Replace $\ell$'s contents by $v$'s contents and replace $p$'s contents by the
contents of $v'$.  Delete $v$ and $v'$.

\subsection{Reducing the Number of Nodes}
\label{sec:treesizereduction}
The tree can be modified to have each leaf node correspond to a set of
$\Theta(s')$ points to reduce the additional space used for maintaining this
tree (pointers and such).  Recall that $s'$ is the threshold for computing a
coreset.  To reduce the number of nodes in the tree this way, we maintain the
invariant that each leaf node, except possibly one, contains a set of size
$s_\ell$ with $s'/2 \le s_\ell \le s'$.  To maintain this invariant, we use a
pointer $p_s$ that points to a leaf node with less than $s'/2$ elements if such
a leaf node exists.

Whenever a point is inserted, we add it to the leaf node, say $\ell_e$ pointed
to by $p_s$.  If $\ell_e$ now contains at least $s'/2$ points, then we make
$p_s$ a null pointer.  If $p_s$ was a null pointer already, then we create a new
leaf node, say $\ell_n$, insert the new point in $\ell_n$, and make $p_s$ point
to $\ell_n$.  The new leaf node $\ell_n$ is inserted in the tree as described in
Section~\ref{sec:heap_description}.

Whenever a point is deleted, we check if the leaf node, say $\ell_d$ that
contains it now contains less than $s'/2$ points.  If $\ell_d$ contains less
than $s'/2$ points, and $p_s$ points to some leaf node, say $\ell_e$, then we
move points in $\ell_d$ into $\ell_e$ and delete $\ell_d$.  (Deletion of a leaf
node is handled as described in Section~\ref{sec:heap_description}.)  If $p_s$
does not point to any leaf node, then we make it point to $\ell_d$.

As usual, we recompute all nodes on the affected leaf-to-root path.


\section{Lower Bounds}
\label{sec:lower}
In this section, we show lower bounds.  We first see a space lower bound and
then a conditional lower bound on the time per operation.

\subsection{Space Lower Bound}
We show a simple and very general space lower bound.  Consider any problem that
on input $X$ has to output a feasible solution that is a subset of $X$.
Moreover, if $X$ non-empty, then all feasible solutions are also non-empty.
Call such a problem \emph{compliant}.  Clearly, computing any bounded
approximation for $k$-median and $k$-means and the problem of constructing any
coreset with bounded quality are compliant.  To get a linear space lower bound
for fully-dynamic algorithms that solve a compliant problem, we use the
communication problem of $\indx$.  In $\indx_N$, Alice's input is an $N$-bit
string and Bob's input is an index $I \in \{1, 2, \ldots, N\}$.  Alice sends one
message to Bob, and he needs to correctly output the bit at position $I$.  By a
well-known communication complexity lower bound, Alice must send a message of
size $(1-H_2(3/4))N \ge 2N/11$ bits so that Bob can correctly output with a
success probability of $3/4$; here $H_2$ is the binary entropy function.

\begin{theorem}
  \label{thm:spacelb}
  A fully-dynamic algorithm for a compliant problem that works in the presence
  of an adaptive adversary and has success probability $1 - 1/(8n^2)$ must use
  space $\Omega(n)$, where $n$ is the current input size.
\end{theorem}
\begin{proof}
  We describe the reduction for any compliant problem in a metric space, such as
  $1$-median or $1$-means, but it can be naturally generalized to any compliant
  problem.  Alice defines
  \[
    X = \{j : j\text{th bit in her string } = 1\}\,,
  \]
  and
  distance between any two points of $X$ to be $1$.  She runs the
  fully-dynamic algorithm on $X$ and sends the memory snapshot to Bob.  Bob
  queries for a solution and if $X$ is nonempty, a nonempty solution $S_1$ would
  be returned.  He deletes the points in $S_1$ and queries again to get $S_2$,
  and so on until $\emptyset$ is returned.  There would be at most $N$ such
  queries.  Note that this works because the algorithm works under an adaptive
  adversary.  If one of the $S_\ell$s in this process contains $I$, which is
  Bob's input for the index problem, then Bob outputs $1$, else he outputs $0$.
  In the worst case, Bob makes $N$ queries, where query number $i$ would have
  failure probability at most $1/(8(N - i + 1)^2)$.  So overall failure
  probability by the union bound is at most
  \[
    \sum_{i = 1}^N \frac{1}{8 (N - i + 1)^2} \le \frac{1}{8} \sum_{i = 1}^\infty
    \frac{1}{i^2} = \frac{1}{8} \frac{\pi^2}{6} \le \frac{1}{4}\,.
  \]
  Alice communicated as many bits as the space usage of the dynamic algorithm.
  Then, by the $\indx_N$ lower bound, the space usage of the algorithm is at
  least $2N/11 \ge 2n/11$ bits.
\end{proof}

\subsection{Conditional Lower Bounds on the Time Per Operation}
Now, we show conditional lower bounds on the time per update and query for
fully-dynamic $k$-means algorithms.  They are based on the
\emph{OMv-conjecture}~\cite{HenzingerKNS15}: You are given an $N\times N$
Boolean matrix $M$ that can be preprocessed in polynomial time.  Then, an online
sequence of $N$-dimensional Boolean vectors $v^1,\dots ,v^N$ is presented and
the task is to compute each $Mv^i$ (using Boolean matrix-vector multiplication) before
seeing the next vector $v^{i+1}$.  The conjecture is that finding all the $N$
answers takes time $\Omega(N^{3-\gamma})$ for any constant $\gamma > 0$.
In~\cite{HenzingerKNS15} also the following \emph{OuMv problem} was presented:
You are given an $N\times N$ Boolean matrix $M$ that can be preprocessed in
polynomial time and an online sequence of Boolean vector pairs
$(u^1, v^1),\dots ,(u^N, v^N)$ with the goal to compute each $(u^i)^T Mv^i$
(using Boolean matrix-vector multiplication) before seeing the next vector
pair $(u^{i+1}, v^{i+1})$. Under the OMv conjecture, finding $N$ answers for the
OuMv problem such that the error probability is at most 1/3 takes time
$\Omega(N^{3-\gamma})$ for any constant $\gamma > 0$.  We will show a reduction from
the latter problem to prove the following result.
\begin{theorem}
  \label{thm:lbomv}
  Let $\gamma > 0$ be a constant.  Under the OMv conjecture, for any $\delta > 0$,
  there does not exist a fully-dynamic algorithm that maintains a
  $(4 - \delta)$-approximation for $k$-means with
  amortized update time $O(k^{1-\gamma})$ and query time $O(k^{2-\gamma})$ such that
  over a polynomial number of updates the error probability is at most $1/3$.
\end{theorem}
\begin{proof}
  For the ease of presentation, we assume that $k$ is even; if $k$ is odd, the
  construction can be easily adapted.  We set $N = k/2$.  Given an OuMv instance
  with $N \times N$ matrix $M$, we construct the following metric space with
  distance function $d$ from it:

  The metric space $U$ consists of $4N$ points numbered from $1$ to $4N$.  For
  any $1 \le i < j \le N$ and $N+1 \le i < j \le 2N$, the distance $d(i,j) = 2$.
  Furthermore, for $1 \le i \le N$ and $N+1 \le j \le 2N$, the distance
  $d(i, j) = 1$ if $M_{i,j-N} = 1$, and $d(i,j) = 2$ otherwise.  Additionally,
  all $2N$ points $2N + 1, \ldots, 4N$ are at distance $100$ from each other and
  from all the other points.

  We use a $k$-means data structure to solve a $u^T M v$ computation as follows:
  Initially the set $X$ is empty.  When given a vector pair $(u,v)$, let $p$ be
  the number of ones in $v$ and in $u$.  Note that $p \le 2N$.  We insert the
  points $i$ such that $u_i =1$ and the points $j$ such that $v_{j-N} = 1$ into
  $X$ and additionally $2N +1 - p$ of the points $\ell$ with $\ell > 2N$.  Thus
  $|X| = 2N +1 = k +1$.  Then we ask a $k$-means query.  Afterwards, we delete
  the inserted points.

  If $u^T Mv = 1$, then there exist indices $i$ and $j$ such that
  $u_i = 1, M_{i,j} = 1, $ and $v_j = 1$.  Consider the optimal solution that
  consists of all points in $X$ except for point $i$.  Note that the cost of
  this solution for the $k$-means problem is $1$.

  If $u^T Mv = 0$, then any optimal solution must also consist of $2N+1 -p$ of
  the points $\ell$ with $\ell > 2N$, and all but one of the other points in
  $X$.  But as none of the points in $X$ has distance smaller than 2 to any
  other point in $X$, the cost of the solution is at least 4 for $k$-means.
  Thus, any $(4 - \delta)$-approximation for $k$-means can distinguish between
  the cases $u^T Mv = 1$ and $u^T Mv = 0$.  Hence, the OMv conjecture implies
  that it takes at least time $\Omega(N^{2-\gamma})$ time to execute the above
  $2N$ update operations and $1$ query operation.  This implies the claimed
  lower bound.
\end{proof}


\section{Acknowledgments}
\label{sec:acknowledgements}
We thank Ali Mehrabi for pointing to us the dynamic coreset algorithm in~\cite{har_peled_mazumdar04}.


\bibliographystyle{alpha}
\bibliography{ref}
\appendix
\section{Proof of Lemma~\ref{lem:binom}}
\label{app:lembinom}
\begingroup
\def\thetheorem{\ref{lem:binom}}
\begin{lemma}
  For any positive integer $\ell$ and $\alpha \in \RR_+$, we have
  \[
    \alpha + (1 + \alpha)\sum_{i = 1}^{\ell - 1}\binom{\ell - 1}{i} \alpha^i =
    \sum_{i = 1}^{\ell}\binom{\ell}{i} \alpha^i\,.
  \]
\end{lemma}
\begin{proof}
  \begingroup
  \allowdisplaybreaks
  \begin{align*}
    \alpha + (1 + \alpha)\sum_{i = 1}^{\ell - 1}\binom{\ell - 1}{i} \alpha^i
    & = \alpha + \sum_{i = 1}^{\ell - 1}\binom{\ell - 1}{i} \alpha^i + \sum_{i = 1}^{\ell - 1}\binom{\ell - 1}{i} \alpha^{i + 1}\\
    & = \binom{\ell - 1}{0}\alpha + \sum_{i = 1}^{\ell - 1}\binom{\ell - 1}{i} \alpha^i + \sum_{i = 1}^{\ell - 1}\binom{\ell - 1}{i} \alpha^{i + 1}
      \hspace{0.5cm}\text{using the fact $\binom{\ell - 1}{0} = 1$}\\
    & = \binom{\ell - 1}{0}\alpha + \sum_{i = 1}^{\ell - 1}\binom{\ell - 1}{i} \alpha^i + \sum_{i = 2}^{\ell}\binom{\ell - 1}{i - 1} \alpha^i\\
    & \hspace{2cm}\text{change of index in the second summation}\\
    & = \sum_{i = 1}^{\ell - 1}\binom{\ell - 1}{i} \alpha^i + \sum_{i = 1}^{\ell}\binom{\ell - 1}{i - 1} \alpha^i\\
    & \hspace{2cm}\text{incorporating the first term in the second summation}\\
    & = \sum_{i = 1}^{\ell}\binom{\ell - 1}{i} \alpha^i + \sum_{i = 1}^{\ell}\binom{\ell - 1}{i - 1} \alpha^i
      \hspace{0.5cm}\text{using the fact $\binom{\ell - 1}{\ell} = 0$}\\
    & = \sum_{i = 1}^{\ell}\left(\binom{\ell - 1}{i} + \binom{\ell - 1}{i - 1} \right)\alpha^i\\
    & = \sum_{i = 1}^{\ell}\binom{\ell}{i} \alpha^i\,,
  \end{align*}
  where we use $\binom{\ell}{i} = \binom{\ell - 1}{i} + \binom{\ell - 1}{i - 1}$
  in the last step.
  \endgroup
\end{proof}
\addtocounter{theorem}{-1}
\endgroup



\end{document}